\documentclass{article}
\usepackage[T1]{fontenc}
\usepackage[latin9]{inputenc}
\usepackage{geometry}
\usepackage[numbers]{natbib}
\usepackage[unicode=true,pdfusetitle,
 bookmarks=true,bookmarksnumbered=false,bookmarksopen=false,
 breaklinks=false,pdfborder={0 0 0},pdfborderstyle={},backref=false,colorlinks=true]
 {hyperref}
\hypersetup{
 linkcolor=blue,citecolor=blue}

\usepackage{hyperref}

\usepackage{float}
\usepackage{amsmath}
\usepackage{amsthm}
\usepackage{amssymb}
\usepackage{graphicx}

\makeatletter

\providecommand{\tabularnewline}{\\}
\floatstyle{ruled}
\newfloat{algorithm}{tbp}{loa}
\providecommand{\algorithmname}{Algorithm}
\floatname{algorithm}{\protect\algorithmname}

\theoremstyle{remark}

\theoremstyle{plain}
\newtheorem{prop}{\protect\propositionname}
\theoremstyle{definition}
 \newtheorem{example}{\protect\examplename}

\makeatother

\usepackage{babel}
\providecommand{\examplename}{Example}
\providecommand{\propositionname}{Proposition}
\providecommand{\remarkname}{Remark}

\begin{document}
\title{Penalised t-walk MCMC}
\author{Felipe J Medina-Aguayo, J Andr\'es Christen}
\date{Centro de Investigaci\'on en Matem\'aticas (CIMAT), A.C. \\ Jalisco S/N, Col. Valenciana, CP: 36023, Guanajuato, Gto, M\'exico}
\maketitle

\begin{abstract}
Handling multimodality that commonly arises from complicated statistical models remains a challenge. Current Markov chain Monte Carlo (MCMC) methodology tackling this subject is based on an ensemble of chains targeting a product of power-tempered distributions. Despite the theoretical validity of such methods, practical implementations typically suffer from bad mixing and slow convergence due to the high-computation cost involved. In this work we study novel extensions of the t-walk algorithm, an existing MCMC method that is inexpensive and invariant to affine transformations of the state space, for dealing with multimodal distributions. We acknowledge that the effectiveness of the new method will be problem dependent and might struggle in complex scenarios; for such cases we propose a post-processing technique based on pseudo-marginal theory for combining isolated samples.

\end{abstract}

\section{Introduction}

Performing Bayesian analyses has become routine in many areas of application and research and, in turn, these are usually carried out using Monte Carlo methodology \citep[see e.g.][]{RobertNCasella_2013}. A popular and effective method, belonging to this class of methods, is Markov Chain Monte Carlo (MCMC), which is based on the construction of a Markov chain that converges to the desired target or posterior distribution \citep[see][for example]{Gamerman2015}.

However, naive implementations of MCMC may result in chains that require a prohibitively large number of iterations in order to be of any use. For example, implementing the celebrated Metropolis--Hastings (MH) algorithm \citep{Metropolis1953,Hastings1970} using highly local proposals may result in chains that do not explore the target fully in a reasonable amount of time. Thus, constructing informative proposals is essential for obtaining useful samplers. This may be achieved by incorporating gradient information from the target, as done by Langevin and Hamiltonian algorithms \citep{Roberts1996,Neal2011a}. Another possibility, that avoids computing derivatives, is adaptive MCMC \citep{Haario2001,Roberts2009} aiming at learning the shape of the target in an online manner and modifying the proposal accordingly. Other approaches rely solely on extending the state space, and consequently extending the target distribution, for constructing good proposals \citep{Foreman-Mackey2013a}. An other example of the latter is the t-walk algorithm \citep{Christen2010}, which will be our subject of study.  It is based on 4 different moves that are invariant under affine transformations \citep[the ``traverse'', the ``walk'', the ``hop'' and the ``blow'' moves, see][for details]{Christen2010}.  The twalk is implemented in an oficial R package (rtwalk), in Python, Matlab, C and C++, \url{https://www.cimat.mx/~jac/twalk/}.

Nevertheless, and similarly to other MCMC methods, in the presence
of multimodality the t-walk may become trapped in a subset of modes,
especially when two or more modes are separated by large regions (valleys)
of very small probability. The chain becomes trapped since
the proposals are not bold enough for jumping into regions that are
far relative to where the chain currently lives. In this respect, in the Bayesian perspective, naively one may argue that multimodal posteriors are merely the result of a bad modelling process or due to a poor choice of priors. However, multimodality is not an artificial phenomenon, resulting from poor decisions by the user, since it commonly arises and is unavoidable in many well studied Bayesian problems.  In particular, inference involving non-linear regressors, as in Bayesian inverse problems, may result in complex multimodal posteriors. Some examples inlcude the FitzHugh-Nagumo model \citep{FitzHugh1961,Nagumo1962}, a soil organic matter decomposition model from \citep{Sierra2012}, or the epidemiological SIR model describing a black plague outbreak in 1666 \citep{Massad2004,JonoskaStojkova2019}. There is indeed a need to develop better MCMC methods to address multimodality in posterior distributions arising from modern Bayesian inferences. 

Currently, the state-of-the-art method for dealing with multimodality is parallel tempering (PT) \citep[see][for example]{Swendsen1986,Geyer1991}.  PT is based on several interacting chains that converge to a product of tempered distributions, one of which is the desired posterior. The chains interact via a series of swap moves that in principle help exploring multiple modes. The challenge, however, is the computational cost involved in the algorithm as one usually requires a large number of tempered distributions and iterations \citep{Woodard2009a} for the method to work correctly. Hence, this and similar approaches might not deal well with multimodality as the computational expense will render the approach unfeasible.  The t-walk can sample from posteriors with several modes with different scales \citep[see][fig. 2 and 3]{Christen2010} but, as with most other methods, once these modes are separated by areas of very low density, then it may get trapped in one particular region of the sample space.
Various staring points should always be used in the practice of MCMC, increasing the chance of finding how and where  our chain got trapped in different modes.  How to join these samples from different regions into one that has the correct target is not obvious.

The t-walk has proved its value and has been used in several studies \cite{Capistran2016b,Ward2012,Aquino-Lopez2018,Capistran2012,Villa2018,Rubio2014,Rubio2016,Rubio2017,Rubio2015} and is currently used routinely by several users in some data analysis software in paleoecology \cite{Blaauw2011} and recently in the COVID19 epidemic modeling \citep{Capistran2020,Acuna-Zegarra2020}.  Certainly, it is not redundant to improve the performance of the t-walk in the presence of multimodal targets. Moreover, the ideas presented here could in principle be imported to improve the performance of other MCMC algorithms.

\subsection*{Outline}

In this work we propose studying the t-walk algorithm in the presence of multimodality by including an additional 5-th move to the standard 4 moves already mentioned.  This additional move is constructed using a ``penalised proposal''. The idea is to construct proposals, easy to sample from, that penalise a region where the chain currently stands and that, ideally, favour regions with meaningful posterior mass that are located further away. Such a penalisation is developed in Section \ref{sec:PenProp} and followed by some illustrative examples in Section \ref{sec:examples}. In case the penalisation fails, in Section \ref{sec:combining} our aim is to combine two different chains, targeting the same posterior, but that were unable to jump between modes. That is, they got stuck in two separate regions of the state space. Hence, these pseudo-samples are distributed according to different sub-posteriors that need to be combined in a particular way for obtaining a true sample from the target. The procedure we propose is based on pseudo-marginal methodology \citep{Andrieu2009} which introduces unbiased estimators of intractable quantities within an MCMC algorithm. Finally, in Section \ref{sec:discussion} we provide a final discussion and possible extensions of this work.

\section{Constructing penalised proposals}\label{sec:PenProp}

\subsection{MCMC and the t-walk}

Our interest is to describe a posterior or target distribution $\pi$
on some measurable space $\left(\mathcal{X},\mathcal{B}\left(\mathcal{X}\right)\right)$,
 $\mathcal{X}\subseteq\mathbb{R}^{d}$. We also
assume throughout that the distributions presented have densities
with respect to the Lebesgue measure. As it commonly arises in Bayesian inference, the
target density has the following generic form
\begin{align*}
\pi\left(x\right) & =\frac{\gamma\text{\ensuremath{\left(x\right)}}}{Z},\quad\text{for }x\in\mathcal{X},
\end{align*}
where $\gamma$ is the unnormalised version of $\pi$ and where $Z=\int_{\mathcal{X}}\gamma\left(dx\right)$
is the intractable normalising constant.

An effective method for describing the posterior $\text{\ensuremath{\pi}}$
is through simulation using Markov chain Monte Carlo (MCMC) methods.
These methods work by constructing a Markov chain $\left(X_{t}\right)_{t\ge0}$,
on $\mathcal{X}$, such that the target $\pi$ is the limiting stationary
distribution. Thus, expectations of the form 
\begin{align*}
\pi\left(h\right) & :=\mathbb{E}_{X\sim\pi}\left[h\left(X\right)\right]=\int_{\mathcal{X}}h\left(x\right)\pi\left(dx\right),\quad\text{for }h:\mathcal{X}\rightarrow\mathbb{R},
\end{align*}
may be estimated, under suitable conditions of aperiodicity and irreducibility
for the chain \citep[see e.g.][]{Roberts2004}, via ergodic averages
\begin{align*}
\widehat{\pi}_{b,T}\left(h\right)= & \frac{1}{T}\sum_{t=b+1}^{b+T}h\left(X_{t}\right),\quad\text{for }b\geq0,T\geq1.
\end{align*}
The integer $b$ in the previous equation defines the ``burn-in period''
of the chain, meaning that the portion $\left(X_{t}\right)_{t\leq b}$
is discarded when estimating $\pi\left(h\right)$, as a pseudo-transient stage of the chain.

The t-walk algorithm introduced in \cite{Christen2010} belongs to the aforementioned
class of MCMC methods, which differs to standard implementations by
exploring an extended target distribution $\bar{\pi}$, defined on
$\left(\mathcal{X}^{2},\mathcal{B}\left(\mathcal{X}^{2}\right)\right)$,
and with density given by
\begin{align*}
\text{\ensuremath{\bar{\pi}\left(x,y\right)}}= & \pi\left(x\right)\pi\left(y\right),\quad\text{for }\left(x,y\right)\in\mathcal{X}^{2}.
\end{align*}
Consequently, the resulting chain $\left(X_{t},Y_{t}\right)_{t\geq0}$
takes values on $\mathcal{X}^{2}$ and converges marginally towards
$\pi$ in the sense that either $\left(X_{t}\right)_{t}$ or $\left(Y_{t}\right)_t$
converges to $\pi$. In this case, working on a larger space is useful
as one is able to construct informative proposals (combining both points, that is proposals of the kind $q( \cdot | X_t, Y_t)$ or $q( \cdot | X_t, Y_t)$) which are key for
exploring efficiently the state space and thus for converging rapidly
towards $\pi$.

At the moment, the 4 different moves implemented in the t-walk algorithm may not be able to overcome the complexity that arises from a target having multiple modes. For such cases we are proposing a 5-th move, which we have termed the ``penalty'' move, that proposes values from regions that are relatively distant to the current state of the chain $(X_t,Y_t)$. As the name suggests, this is achieved by penalising a neighbourhood that is close to $(X_t,Y_t)$, which creates a flattening effect on some standard density from a distribution that is easy to sample from. The right plot of Figure \ref{fig:Fig0_penalty} shows an example of this effect.

Therefore, we introduce a family of penalty functions
$\left\{ \varphi_{xy}\right\} _{\left(x,y\right)\in\mathcal{X}^{2}}$
such that $\varphi_{xy}:\mathcal{X}\rightarrow\left[0,1\right]$ modifies
any given proposal in such way that sampling near a region informed
by $\left(x,y\right)$ is not very likely. For instance, assuming
the current state of the t-walk chain is $\left(X_t,Y_t\right)=\left(x,y\right)$ the penalised
version of a proposal distribution $g\left(\cdot \mid x,y \right)$ defined on $\left(\mathcal{X},\mathcal{B}\left(\mathcal{X}\right)\right)$
will have a density
\begin{align}
\breve{g}\left(w \mid x,y \right) & \propto g\left(w \mid x,y \right)\varphi_{xy}\left(w\right),\text{\ensuremath{\quad\text{for }w\in\mathcal{X}}.}\label{eq:penalised_proposal}
\end{align}

Notice that the t-walk chain lies on $\mathcal{X}^{2}$ whereas the proposals
$g$ and $\breve{g}$ are defined on $\mathcal{X}$. This apparent problem is overcome in the following section where we describe a simple transformation for obtaining two components, each belonging to $\mathcal{X}$, out of a single draw from $\breve{g}$. In addition, it is worth mentioning that the choice of $g$ is in principle arbitrary; in fact one could consider penalising one of the four moves in the t-walk. However, we do not follow this approach as the existing moves in the t-walk algorithm only update one of the two components at a time (either $x$ or $y$), which could be problematic since components may get stuck in different modes. 

Instead, for reasons that will become clear in the following sections, our choice for $g$ will be simply a symmetric distribution that is easy to sample from, e.g. a Gaussian or t-distribution. Moreover, we will discuss some feasible approaches for defining and sampling from $\breve{g}$, both using the gradient of $\log \pi$ and not. We finish this section introducing possible choices for $\varphi_{xy}$.  

The general idea of the penalty function $\varphi_{xy}$ is that at its centre, say $\mu_{xy}$, the value of resulting penalised density $\breve{g}$ will be exactly zero, and as we start moving away from $\mu_{xy}$ the value of the density will increase. As an example we could consider flipping and renormalising a Gaussian density, i.e.
\begin{align*}
\varphi_{xy}\left(w\right)=1-\exp\left(-\frac{1}{2}\left(w-\mu_{xy}\right)^{T}\Sigma_{xy}^{-1}\left(w-\mu_{xy}\right)\right), 
\end{align*}
for some $\Sigma_{xy}$. Notice that $\varphi_{xy}\in[0,1)$ and when multiplied by some $g$ the resulting penalised proposal $\breve{g}$ will be flattened around $\mu_{xy}$, as shown in Figure~\ref{fig:Fig0_penalty}.

\begin{figure}[!ht]
\begin{centering}
\includegraphics[width=0.48\linewidth]{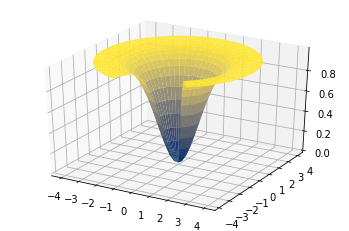}\includegraphics[width=0.48\linewidth]{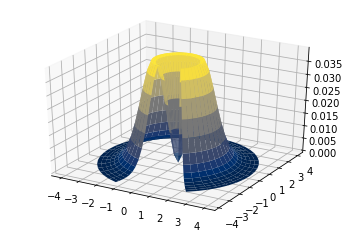}
\par\end{centering}
\caption{Flipped Gaussian penalty function $\varphi_{xy}$ (left) and resulting penalised Gaussian density $\breve{g}$ (right), both centred at the origin.}\label{fig:Fig0_penalty}
\end{figure}

Examples of other penalties include the following continuously differentiable functions:
\begin{itemize}
\item (Flipped t-density)
\begin{align*}
\varphi_{xy}\left(w\right)=1-1/\left(1+\left(w-\mu_{xy}\right)^{T}\Sigma_{xy}^{-1}\left(w-\mu_{xy}\right)\right);    
\end{align*}
\item (Flipped bump function) 
\begin{align*}
\varphi_{xy}\left(w\right)&=1-\exp\left(1-1/\left(1-\left(w-\mu_{xy}\right)^{T}\Sigma_{xy}^{-1}\left(w-\mu_{xy}\right)\right)\right)\\
&\quad \times \mathbf{1}\left(\left(w-\mu_{xy}\right)^{T}\Sigma_{xy}^{-1}\left(w-\mu_{xy}\right)<1\right);    
\end{align*}
with $\mathbf{1}\left(A\right)$ denoting the indicator function
on the event $A$;
\end{itemize}
where $\mu_{xy}$ and $\Sigma_{xy}$ are location and
scale parameters determined by the current state $\left(x,y\right)$.
In general terms we will restrict to penalty functions $\varphi_{xy}$
of the form
\begin{align}
\varphi_{xy}\left(w\right) & =1-\rho_{0}^{-1}\text{\ensuremath{\rho}}\left(\Sigma_{xy}^{-1/2}\left(w-\mu_{xy}\right)\right),\label{eq:penalty_locScale}
\end{align}
where $\rho$ is some standard symetric distribution and therefore $\rho_{0}^{-1}\text{\ensuremath{\rho}}\left(\Sigma_{xy}^{-1/2}\left(w-\mu_{xy}\right)\right)$ is a location-scale family;  $\rho_{0}=\ensuremath{\rho\left(0\right)}$.

\subsection{Penalty using gradient information}\label{sec:penalty}

We construct a density $\breve{g}$ that we can easily simulate
from and that is informed by the gradient of $\log\pi$.
To define $g$
consider
a symmetric distribution $g\left(\cdot \mid x,y \right)$ centred at $\tilde{\mu}_{xy}$
with dispersion parameter $\tilde{\Sigma}_{xy}^{2}$.
For example, a Gaussian or t-distribution with central parameter $\tilde{\mu}_{xy}=(x+y)/2$
and dispersion $\tilde{\Sigma}_{xy}=\text{diag}(|x-y|^2)$. Notice that the previous
parameters may be different to those from the penalty function $\varphi_{xy}$
introduced in the previous section. In any case, we assume that it is computationally straightforward to simulate from $g$.

The following proposition is an extension of \cite[][Proposition 3.1]{Livingstone2019} and will be useful for constructing, and simulating from,  $\text{\ensuremath{\breve{g}}}_{xy}$.

\begin{prop}
Suppose $\mathcal{X}\subseteq \mathbb{R}^d$, and assume that a penalty function $\tilde{\varphi}_{xy}$ satisfies,
for some quantity $\tilde{\mu}_{xy}\in\mathcal{X}$,
\begin{align}
\tilde{\varphi}_{xy}\left(\tilde{\mu}_{xy}-w\right) & =1-\tilde{\varphi}_{xy}\left(\tilde{\mu}_{xy}+w\right),\quad w\in\mathcal{X},\label{eq:symm_condn}
\end{align}
then for any symmetric distribution $g\left(\cdot \mid x,y \right)$ around $\tilde{\mu}_{xy}$, and defined on $\mathbb{R}^d$
\begin{align*}
\int_{\mathbb{R}^d}g\left(w \mid x,y \right)\tilde{\varphi}_{xy}\left(w\right)dw & =\frac{1}{2}.
\end{align*}
\end{prop}
\begin{proof}
For simplicity we omit the dependence from $(x,y)$ throughout the proof. First notice that
\begin{align*}
\int_{\mathbb{R}^d}g\left(w\right)\tilde{\varphi}\left(w\right)dw = \int_{\mathbb{R}^d} g\left(u + \tilde{\mu} \right)\tilde{\varphi}\left(u+ \tilde{\mu}\right)du.
\end{align*}

Consider now one of the $2^d$ hyperoctants defined on $\mathbb{R}^d$ and, without loss of generality, take the hyperoctant $\mathbb{R}^{+d}$ where all the components are positive. Using \eqref{eq:symm_condn} and the symmetry of $g$, the above integral restricted to $\mathbb{R}^{+d}$ becomes
\begin{align*}
\int_{\mathbb{R}^{+d}}& g\left(u + \tilde{\mu}\right)\tilde{\varphi}\left(u+ \tilde{\mu}\right)du =\\
    &= \int_{\mathbb{R}^{+d}} g\left(u + \tilde{\mu} \right)du - \int_{\mathbb{R}^{+d}} g\left(u + \tilde{\mu} \right) \tilde{\varphi}\left(\tilde{\mu} - u \right)du \\
    &= \int_{\mathbb{R}^{+d}} g\left(u + \tilde{\mu} \right)du - \int_{\mathbb{R}^{+d}} g\left( \tilde{\mu} -u \right) \tilde{\varphi}\left(\tilde{\mu} - u \right)du \\
    &= \int_{\mathbb{R}^{+d}} g\left(u + \tilde{\mu} \right)du - \int_{\mathbb{R}^{-d}} g\left( \tilde{\mu} + u \right) \tilde{\varphi}\left(\tilde{\mu} + u \right)du,
\end{align*}
where $\mathbb{R}^{-d}$ is the opposite hyperoctant of $\mathbb{R}^{+d}$, i.e. we made the change of variable $u$ to $-u$. Therefore, 
\begin{align*}
\int_{\mathbb{R}^{+d}\bigcup \mathbb{R}^{-d}}& g\left(u + \tilde{\mu} \right)\tilde{\varphi}\left(u+ \tilde{\mu}\right)du = \int_{\mathbb{R}^{+d}} g\left(u + \tilde{\mu} \right)du.
\end{align*}

Following the same logic for the remaining $2^d-1$ hyperoctants, and adding all integrals to integrate over the whole space we obtain
\begin{align*}
2 \int_{\mathbb{R}^{d}} & g\left(u + \tilde{\mu} \right)\tilde{\varphi}\left(u+ \tilde{\mu}\right)du = \int_{\mathbb{R}^{d}} g\left(u + \tilde{\mu} \right)du = 1 .
\end{align*}
\end{proof}

A direct implication from the previous result is that 
\begin{align*}
\breve{g}\left(w \mid x,y\right)=2 g\left(w \mid x,y\right)\tilde{\varphi}_{xy}\left(w\right),   
\end{align*}
which does not involve an intractable normalising constant; moreover, one can simulate
from the associated distribution easily using Algorithm \ref{alg:GradPenProp}, as proven in \citep[][Proposition 3.2]{Livingstone2019}. In order to see why this is a valid procedure, suppose the centre $\tilde{\mu}_{xy}=0$, then returning a value $w$ as proposed by $g$ occurs with probability $\tilde{\varphi}_{xy}(w)$. In contrast, according the algorithm the sign of $w$ will be flipped with probability $1-\tilde{\varphi}_{xy}(w)=\tilde{\varphi}_{xy}(-w)$, but (due to symmetry of $g$) proposing $-w$ is equally as likely as proposing $w$. Hence, instead of discarding $w$ after not being accepted, the value can be ``reused'' by returning $-w$ which had the same probability of being proposed.

\begin{algorithm}
INPUT: Current value for the chain $\left(x,y\right)\in\mathcal{X}^2$ and centre $\tilde{\mu}_{xy}$.

OUTPUT: Draw from $\breve{g}$ involving the penalty $\tilde{\varphi}_{xy}$.

\begin{enumerate}
\item Simulate and store $w\leftarrow W\sim g\left(\cdot \mid x,y\right)$ and
$\omega \leftarrow \Omega \sim Unif\left(0,1\right)$;
\item If $\omega \leq \tilde{\varphi}_{xy}\left(w\right)$ return $w$;

otherwise return $2\tilde{\mu}_{xy}-w$.
\end{enumerate}
\caption{Penalised proposal sampler using $\tilde{\varphi}_{xy}$}\label{alg:GradPenProp}
\end{algorithm}

The problem, however, is still how to choose a suitable $\tilde{\varphi}_{xy}$. Leaving the formality aside for a moment, suppose we could sample from the reciprocal of $\pi$, i.e. a distribution with density proportional to $1/\pi(x)$. Incorporating such a move into an MCMC sampler could be useful as one would be able to escape from the attraction of a local mode. We perform this idea, in an approximate manner, by defining the penalty 
\begin{align}
\tilde{\varphi}_{xy}\left(w\right) & =\frac{1}{1+\exp\left(\nabla l\left(\tilde{\mu}_{xy}\right)\cdot\left(w-\tilde{\mu}_{xy}\right)\right)},\label{eq:penalty_tilde}
\end{align}
where $l\left(u\right):=\log\left(\pi\left(u\right)\right)$. The choice in \eqref{eq:penalty_tilde} approximates Barker's acceptance probability since, using a Taylor approximation,
\begin{align*}
\tilde{\varphi}_{xy}\left(w\right) & \approx\frac{1}{1+\exp\left(l\left(w\right)-l\left(\tilde{\mu}_{xy}\right)\right)}=\left(1+\left(\frac{1/\pi\left(w\right)}{1/\pi\left(\tilde{\mu}_{xy}\right)}\right)^{-1}\right)^{-1}.
\end{align*}
Hence $\tilde{\varphi}_{xy}$ approximates the acceptance probability of an MCMC sampler targeting the reciprocal of $\pi$. Finally, notice that \eqref{eq:penalty_tilde} also satisfies the symmetry condition in \eqref{eq:symm_condn}; therefore, a sample from the resulting proposal $\breve{g}$ could be seen as a crude approximation from the reciprocal of $\pi$.

Having simulated $W$ in Algorithm 1, we can propose new values $\left(U,V\right)$ for the chain using either
\begin{align}\label{eq:transfUV_1}
U & =x+(W-\tilde{\mu}_{xy}),\quad V =y+(W-\tilde{\mu}_{xy}),
\end{align}
or
\begin{align}\label{eq:transfUV_2}
U & =y+(W-\tilde{\mu}_{xy}),\quad V =x+(W-\tilde{\mu}_{xy}).
\end{align}
We now present a simple toy example that illustrates the advantage of introducing a penalised move for a target with very-well separated modes.

\begin{example}
Consider a 2-dimensional Gaussian mixture target with 2 components and with density given by
\begin{align*}
\pi\left(x_{1},x_{2}\right) & =0.5\mathcal{N}\left(\left(x_{1},x_{2}\right)\mid\mu_{\text{1}},\Sigma_{1}\right)+0.5\mathcal{N}\left(\left(x_{1},x_{2}\right)\mid\mu_{\text{2}},\Sigma_{2}\right),
\end{align*}
where $\mu_{1}=\left(0,0\right)$, $\mu_{2}=\left(20,-20\right)$ and
\begin{align*}
\Sigma_{1}=\left(\begin{array}{cc}
1 & 0.1\\
0.1 & 1
\end{array}\right), \quad \Sigma_{2}=\left(\begin{array}{cc}
4^{2} & 0.8\times4\times5\\
0.8\times4\times5 & 5^{2}
\end{array}\right).
\end{align*}

The target involves two very different components, one with low and another with high correlation. Exploring these components separately would not represent a challenge for the t-walk algorithm, however the relatively large separation can be problematic. Figure \ref{fig:Fig1_2d2Mix} shows trace plots of the t-walk algorithm (left plot) and the penalised version using the gradient described above (right plot). Notice that the t-walk chain is 10 times longer (500 thousand iterations) than the penalised version; the reason for this was to take into account that the latter is more computationally
expensive due to gradient computations. Nonetheless, observe how the
t-walk only jumps to the second Gaussian component once, whereas
the penalised version visits the two modes regularly.

\begin{figure}[!ht]
\begin{centering}
\includegraphics[width=0.45\linewidth]{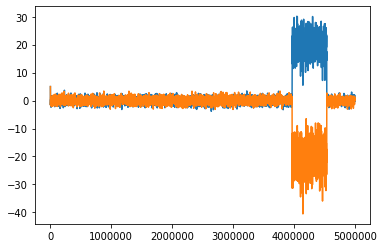}\includegraphics[width=0.45\linewidth]{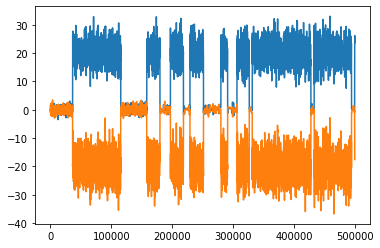}
\par\end{centering}
\caption{Trace plots of the 2 components in Example 1 for the t-walk (left column) and the penalised t-walk (right column).}\label{fig:Fig1_2d2Mix}
\end{figure}

\end{example}

We now turn to a penalised version which in essence is much simpler but that will not require the computation of gradients.

%

\subsection{Penalty without a gradient}

Having gradient information is beneficial but not widely applicable,
e.g. Bayesian inverse problems commonly involve numerical solutions to differential equations, thus obtaining gradients with respect to fixed parameters is not straightforward. In this section we explore much simpler penalised proposals that are easy to evaluate and to sample from.

The idea is to propose a point that is distant to the current location of the chain which is summarised by the single point $\mu_{xy}\in\mathcal{X}$, e.g. $\mu_{xy}=(x+y)/2$.  We consider once more a penalised proposal as in (\ref{eq:penalised_proposal})
with explicit density
\begin{align*}
\breve{g}\left(w \mid x,y\right) & =\frac{g\left(w \mid x,y\right) \varphi_{xy}\left(w\right)}{\mathfrak{Z}_{xy}},
\end{align*}
where $\mathfrak{Z}_{xy}=\int_{\mathcal{X}}g\left(dw \mid x,y\right)\varphi_{xy}\left(w\right)$ is the normalization constant and $g$ is an easy to simulate, heavy-tailed distribution, typically a multivariate t-distribution. Simulating then from $\breve{g}$
can proceed via rejection sampling as described in Algorithm \ref{alg:RejSamp}, and due to the fact that $\mathfrak{Z}_{xy}\leq 1$ since $\varphi_{xy} \in (0,1)$ for any $(x,y)\in\mathcal{X}^2$.

\begin{algorithm}
INPUT: Current value for the chain $\left(x,y\right)\in\mathcal{X}^2$.

OUTPUT: Draw from $\breve{g}$ involving the penalty $\varphi_{xy}$.

\begin{enumerate}
\item Simulate and store $w\leftarrow W\sim g\left(\cdot \mid x,y\right)$ and
$\omega \leftarrow \Omega \sim Unif\left(0,1\right)$
\item If $\varphi_{xy}\left(w\right)<\omega$ go back to step 1;

else go to step 3;
\item Return $w$.
\end{enumerate}
\caption{Penalised proposal sampler using $\varphi_{xy}$}\label{alg:RejSamp}
\end{algorithm}

Notice that if $g$ also belongs to a location-scale family 
\begin{equation}
g\left(w \mid x,y\right)=\frac{\text{1}}{\left|\Upsilon_{xy}^{1/2}\right|}g\left(\Upsilon_{xy}^{-1/2}\left(w-\mu_{xy}\right)\right),\label{eq:g_locScale}
\end{equation}
where $\Upsilon_{xy}^{1/2}=\kappa\Sigma_{xy}^{1/2}$ for some $\kappa>0$,
then the normalising constant $\mathfrak{Z}_{xy}$ satisfies
\begin{align}
\mathfrak{Z}_{xy} & =1-\rho_{0}^{-1}\int_{\mathcal{X}}\frac{\text{1}}{\left|\Upsilon_{xy}^{1/2}\right|}g\left(\Upsilon_{xy}^{-1/2}\left(w-\mu_{xy}\right)\right)\text{\ensuremath{\rho}}\left(\Sigma_{xy}^{-1/2}\left(w-\mu_{xy}\right)\right)dw\nonumber \\
 & =1-\rho_{0}^{-1}\int_{\mathcal{X}}g\left(w\right)\text{\ensuremath{\rho}}\left(\Sigma_{xy}^{-1/2}\Upsilon_{xy}^{1/2}w\right)dw\nonumber \\
 & =1-\rho_{0}^{-1}\int_{\mathcal{X}}g\left(w\right)\text{\ensuremath{\rho}}\left(\kappa w\right)dw=:\mathfrak{Z}.\label{eq:normConst}
\end{align}

Regarding Algorithm \ref{alg:RejSamp}
we now compute the acceptance probability within the rejection sampling
step, which will also be independent of the current state $(x,y)$. This is crucial for controlling the extra computational cost of the penalised t-walk algorithm, as the expected number of iterations before terminating Algorithm \ref{alg:RejSamp} is the same irrespective of the where the chain is located.

\begin{prop}
Suppose the chain is at state $\left(x,y\right)$, also assume the
penalty $\text{\ensuremath{\varphi_{xy}}}$ satisfies (\ref{eq:penalty_locScale})
and the proposal $g$ belongs to a location-scale family as stated
in (\ref{eq:g_locScale}). Then the overall acceptance probability
in the rejection sampling from Algorithm \ref{alg:RejSamp} is independent of
$\left(x,y\right)$.
\end{prop}
\begin{proof}
Let $N$ be the number of trials needed for obtaining a sample from
$\breve{g}$ in the rejection sampling procedure. All the variables generated
up tp the $n$-th trial are $U_{1:n}\stackrel{iid}{\sim}Unif\left(0,1\right)$
and $W_{1:n}\stackrel{iid}{\sim}g\left(\cdot \mid x,y\right)$, therefore
\begin{align*}
\mathbb{P}\left[N>n\right] & =\mathbb{P}\left[\varphi_{xy}\left(W_{1}\right)<U_{1},\dots,\varphi_{xy}\left(W_{n}\right)<U_{n}\right]\\
 & =\prod_{i=\text{1}}^{n}\mathbb{P}\left[\varphi_{xy}\left(W_{i}\right)<U_{i}\right]\\
 & =\left(\mathbb{P}\left[\varphi_{xy}\left(W_{1}\right)<U_{1}\right]\right)^{n}\\
 & =:\left(1-p_{xy}\right)^{n}
\end{align*}
which implies $N\sim Geom\left(p_{xy}\right)$. Thus, since $\varphi_{xy}\in\left[0,\text{1}\right]$
\begin{align*}
p_{xy} & =1-\mathbb{E}\left[\mathbb{P}\left[\varphi_{xy}\left(W_{1}\right)<U_{1}\mid W_{1}\right]\right]\\
 & =1-\mathbb{E}\left[1-\varphi_{xy}\left(W_{1}\right)\right]\\
 & =1-\mathbb{E}\left[\rho_{0}^{-1}\text{\ensuremath{\rho}}\left(\Sigma_{xy}^{-1/2}\left(W_{1}-\mu_{xy}\right)\right)\right]\\
 & =1-\rho_{0}^{-1}\int_{\mathcal{X}}g\left(w\right)\text{\ensuremath{\rho}}\left(\kappa w\right)dw=\mathfrak{Z},
\end{align*}
where the last line follows from (\ref{eq:normConst}) due to the
fact that $g$ belongs to a location-scale family. Therefore,
the probability of terminating the rejection sampling procedure at
the $n$-th trial is
\begin{align*}
\mathbb{P}\left[N=n\right] & =\mathfrak{Z}\left(1-\mathfrak{Z}\right)^{n-1},
\end{align*}
which is independent of the current state $\left(x,y\right)$.
\end{proof}

When both $g$ and $\rho$ are $d$-dimensional Gaussian distributions
we have a closed expression for $\mathfrak{Z}$, namely
\begin{align*}
\mathfrak{Z} & =1-\left(2\pi\right)^{-d/2}\int_{\mathcal{X}}\exp\left\{ -0.5w^{T}\mathbb{I}_{d\times d}\left(1+\kappa^{2}\right)w\right\} dw\\
 & =1-\left|\mathbb{I}_{d\times d}\left(1+\kappa^{2}\right)\right|^{-1/2}\\
 & =1-\left(1+\kappa^{2}\right)^{-d/2}.
\end{align*}
Notice that as $d\rightarrow\infty$ then $\mathfrak{Z}\rightarrow1$,
therefore one might want to scale $\kappa$ with the dimension of
$\mathcal{X}$. Table \ref{tab:acceptProbs} shows the estimated acceptance probabilities
in the rejection sampling procedure for different choices of the proposal
$g$ and the distribution $\rho$ involved in the penalty. The default configuration of the penalised t-walk sampler was set to $\kappa=3$, multivariate t-distribution with 2 df for the penalty, and multivariate t-distribution with 1 df for $g$.

\begin{table}
\begin{centering}
\begin{tabular}{|c|c|c|c|}
\hline 
\multicolumn{4}{|c|}{Gaussian penalty}\tabularnewline
\hline 
\hline 
d / $\kappa$ & 2 & 3 & 4\tabularnewline
\hline 
2 & 0.8392 & 0.9148 & 0.94856\tabularnewline
\hline 
4 & 0.9516 & 0.9832 & 0.9925\tabularnewline
\hline 
8 & 0.9897 & 0.9983 & 0.9997\tabularnewline
\hline 
16 & 0.9991 & 0.9999 & $>1-10^{-5}$\tabularnewline
\hline 
\end{tabular}%
\begin{tabular}{|c|c|c|c|}
\hline 
\multicolumn{4}{|c|}{Multiv. t-distn penalty with 2 df}\tabularnewline
\hline 
\hline 
d / $\kappa$ & 2 & 3 & 4\tabularnewline
\hline 
2 & 0.8671 & 0.9275 & 0.9551\tabularnewline
\hline 
4 & 0.9802 & 0.9931 & 0.9970\tabularnewline
\hline 
8 & 0.9993 & 0.9999 & $>1-10^{-5}$\tabularnewline
\hline 
16 & \multicolumn{3}{c|}{$>1-10^{-8}$}\tabularnewline
\hline 
\end{tabular}
\par\end{centering}
\centering{}\caption{Acceptance rates for Gaussian and multivariate-t penalties proposing from a multivariate t-distribution with 1 df.}\label{tab:acceptProbs}
\end{table}

As discussed earlier, having simulated $W \sim \breve{g}\left(\cdot \mid x,y\right)$ using Algorithm \ref{alg:RejSamp}, which can be thought as a new value for the centre $\mu_{xy}$, we can obtain values for $\left(U,V\right)$ through the simple transformation in \eqref{eq:transfUV_1} or \eqref{eq:transfUV_2}. 
Recalling from \eqref{eq:normConst} that the constant $\mathfrak{Z}_{xy}$ does not depend on $(x,y)$, which implies that $\mathfrak{Z}_{xy}/\mathfrak{Z}_{uv}=1$, the MH ratio in the corresponding iteration of the t-walk becomes
\begin{align}\label{eq:ratioMH}
r\left(x,y;u(w),v(w)\right)&=\frac{\gamma\left(u\right)\gamma\left(v\right) g\left(\mu_{xy} \mid u,v \right) \varphi_{uv}\left(\mu_{xy}\right) }{ \gamma\left(x\right)\gamma\left(y\right)g\left(w \mid x,y\right)\varphi_{xy}\left(w\right)}\times\frac{\mathfrak{Z}_{xy}}{\mathfrak{Z}_{uv}} \nonumber\\
&=\frac{\gamma\left(u\right)\gamma\left(v\right) g\left(\mu_{xy} \mid u,v \right) \varphi_{uv}\left(\mu_{xy}\right) }{ \gamma\left(x\right)\gamma\left(y\right)g\left(w \mid x,y\right)\varphi_{xy}\left(w\right)}.
\end{align}

For completeness, and before considering further numerical examples, we present Algorithm \ref{alg:PenalisedMove} which contains a full description of the penalised moved implemented in the t-walk.

\begin{algorithm}
INPUT: Current value for the t-walk chain $\left(x,y\right)\in\mathcal{X}^2$.

OUTPUT: New value for the chain.

\begin{enumerate}
\item Using the location $\mu_{xy}$ and scale $\Sigma_{xy}$, simulate and store $w\leftarrow W\sim \breve{g}\left(\cdot \mid x,y\right)$ (e.g. via Algorithm \ref{alg:GradPenProp} or \ref{alg:RejSamp});
\item Simulate and store $\omega \leftarrow \Omega \sim Unif\left(0,1\right)$;
\item Obtain $(u,v)$ using either \eqref{eq:transfUV_1} or \eqref{eq:transfUV_2};
\item Compute $r\left(x,y;u,v\right)$ from \eqref{eq:ratioMH};
\item If $\omega\leq r\left(x,y;u,v\right)$ return $(u,v)$;

else return $(x,y)$.
\end{enumerate}
\caption{Penalised move in the t-walk}\label{alg:PenalisedMove}
\end{algorithm}

\section{Simulations}\label{sec:examples}

In this section we present several examples that despite being artificial they are very illustrative for assessing the utility of the penalty move introduced to the t-walk algorithm. As mentioned before, all examples consider the default setting of the penalised t-walk sampler with $\mu_{xy}=(x+y)/2$, $\Sigma_{xy}=\text{diag}(|x-y|^2)$, $\kappa=3$, a multivariate t-distribution with 2 df for the penalty $\varphi_{xy}$, and a multivariate t-distribution with 1 df for $g$.

\begin{example}
We consider the same bi-modal target introduced in Example 1. Figure \ref{fig:Fig2_2d2Mix} shows the trace plots, similar to those from Figure \ref{fig:Fig1_2d2Mix}, and the corresponding kernel density estimation (KDE) plots using 5 million iterations for both the original and penalised t-walk. Observe that, similarly to the case with a penalty involving gradients, the penalised t-walk is able to jump often between modes. The estimated integrated autocorrelation time (IAT), which is commonly regarded as a measure of efficiency \cite{Geyer_1992}, was 387 for the original t-walk and 955 for penalised version. These results are clearly misleading since one could conclude wrongly that the t-walk without penalisation is performing much better since its IAT is less than half to the penalised version. In order to properly compare these quantities we would need to run the chains for much longer; nonetheless this simple experiment shows that relying simply on the IAT for assessing the quality of a sample from a multimodal target may lead to incorrect results.

On average the number of penalised moves was set to 10 percent of the total number of iterations. The average number of penalised moves that were accepted was only 3 percent, which produced a decrease in the global acceptance from 33 percent (using the standard t-walk) to 16 percent (when introducing the penalised move).


\end{example}

\begin{figure}[!ht]
\begin{centering}
\includegraphics[width=0.45\linewidth]{./Plots/2d/trace_TWalk}\includegraphics[width=0.45\linewidth]{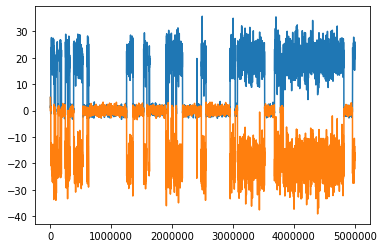}
\par\end{centering}
\begin{centering}
\includegraphics[width=0.45\linewidth]{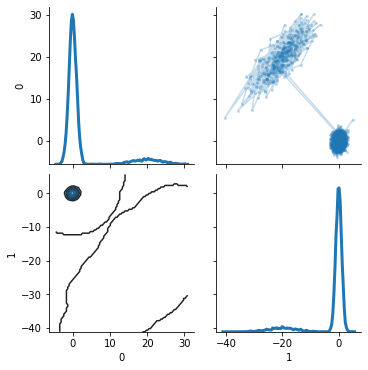}\includegraphics[width=0.45\linewidth]{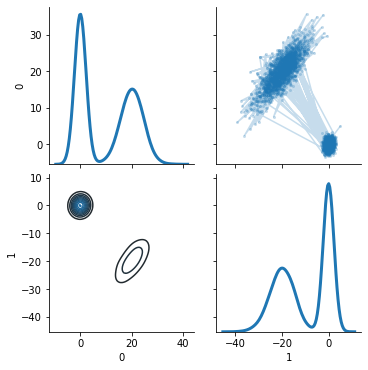}
\par\end{centering}
\begin{centering}
\includegraphics[width=0.45\linewidth]{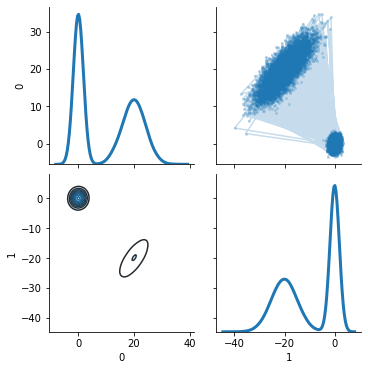}
\par\end{centering}
\caption{Plots from Example 2 for t-walk (left column plots)
and penalised t-walk (right column plots). Top-row figures correspond to trace plots for the 2 components of the chain. Middle-row figures correspond to the KDE plots estimated from the corresponding samples. Bottom figure shows the KDE plots from ground truth by sampling 20 thousand i.i.d. variables.}\label{fig:Fig2_2d2Mix}
\end{figure}

\begin{example}
Now consider a 3-dimensional Gaussian mixture target with 9 components given by
\begin{alignat*}{1}
\pi\left(x_{1:3}\right) & \propto\sum_{i=1}^{8}\mathcal{N}\left(x_{1:3}\mid\mu_{i},\sigma_{i}^{2}\mathbb{I}_{3\times3}\right)+\mathcal{N}\left(x_{1:3}\mid\left(30,30,30\right),10\mathbb{I}_{3\times3}\right),
\end{alignat*}
where each $\mu_{i}$ is of the form $\mu_{i}= 10 \ \left((-1)^{a_{i}},(-1)^{b_{i}},(-1)^{c_{i}}\right)$ for some $a_{i},b_{i},c_{i}\in\left\{ 0,1\right\} $ such that $\left(a_{i},b_{i},c_{i}\right)\neq\left(a_{j},b_{j},c_{j}\right)$
if $i\neq j$. In words, the means for the first 8 components are located at each of the 8 vertices from the cube of length 20 and centred at the origin. Regarding the variance terms, the values for $\left\{ \sigma_{i}^{2} \right\}_i$ range from 0.25 to 10.

Figure \ref{fig:Fig3_3d9Mix} shows once more how the standard t-walk struggles visiting all the different modes. In fact the t-walk gets stuck in a single mode whereas the penalised version is able to visit all 9 modes in the first million iterations. Clearly, the run for the penalised version requires more iterations as the global acceptance rate is only 4 percent. However, despite the unappealing slow convergence shown by the penalised version, the method is able to expose the multimodality issue.



\end{example}
\begin{figure}[!ht]
\begin{centering}
\includegraphics[width=0.45\linewidth]{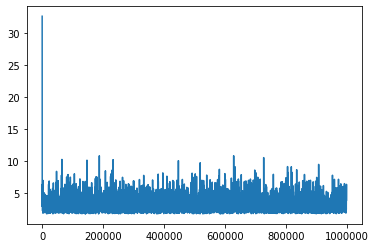}\includegraphics[width=0.45\linewidth]{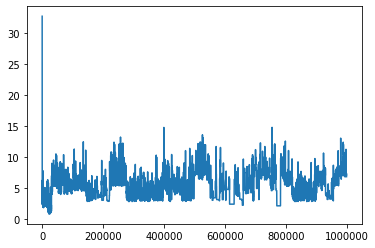}
\par\end{centering}
\begin{centering}
\includegraphics[width=0.45\linewidth]{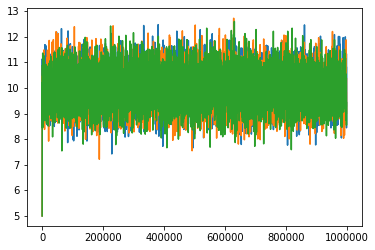}\includegraphics[width=0.45\linewidth]{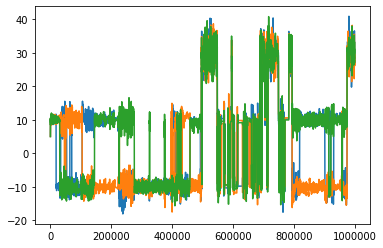}
\par\end{centering}
\begin{centering}
\includegraphics[width=0.45\linewidth]{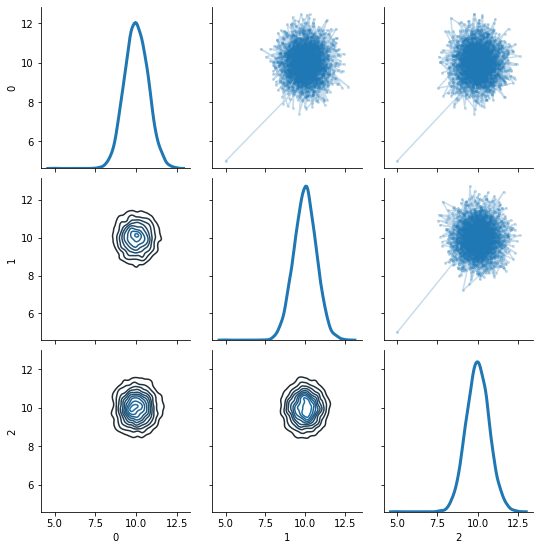}\includegraphics[width=0.45\linewidth]{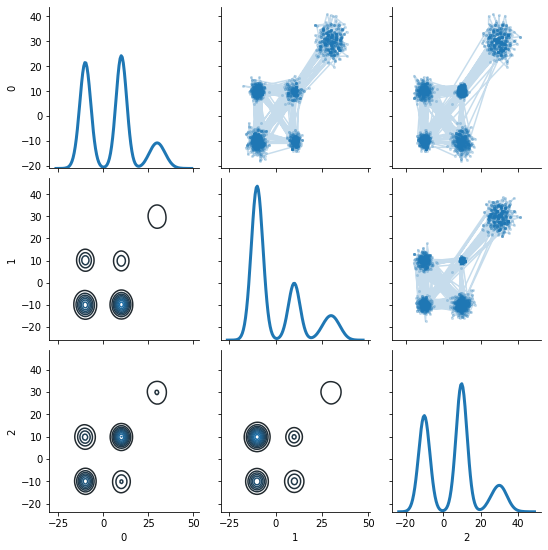}
\par\end{centering}
\caption{Plots from Example 3 for t-walk (left column plots)
and penalised t-walk (right column plots). Top-row figures correspond to trace plots for the value of $-\log(\pi)$. Middle-row plots are the traces for the 3 components of the chain. Bottom-row figures correspond to the KDE plots estimated from the corresponding samples.}\label{fig:Fig3_3d9Mix}
\end{figure}

\begin{example}
For this final example consider a 10-dimensional target from a 3 component mixture of banana-shaped distributions given by
\begin{align*}
\pi\left(x_{1:10}\right) & \propto\sum_{i=1}^{3}\mathcal{N}\left(\phi_{b_{i}}\left(x_{1:10}\right)\mid\mu_{i},\Sigma\right),
\end{align*}
where $\mu_{1}=\left(-3,-3,\dots,-3\right)$, $\mu_{2}=\left(0,0,\dots,0\right)$,
$\mu_{3}=\left(3,3,\dots,3\right)$,
\begin{align*}
    \phi_{b}\left(x_{1:d}\right)=\left(x_{1}+bx_{1}^{2}-100b,x_{2}, \dots,x_{d}\right),
\end{align*}
$b_{1:3}=\left(-0.03,0,0.03\right)$, and $\Sigma=\text{diag}\left(100,1,1,\dots,1\right)$. 

The shape of this type of densities can be very challenging for MCMC samplers as they may get stuck in the tails of the distribution. From Figure \ref{fig:Fig4_10d3Mix} we can observe that the standard t-walk is able to explore the 3 modes with a global acceptance of 13 percent. In contrast, the penalised version showed a global acceptance of 10 percent, with an acceptance for the penalty move of only 0.3 percent. The IAT for the t-walk was 630 versus the marginally lower value of 578 for the penalised version. One can argue that for this example the benefit from the penalised t-walk is minimal, but does not interfere with the t-walk performance.


\end{example}
\begin{figure}[!ht]
\begin{centering}
\includegraphics[width=0.45\linewidth]{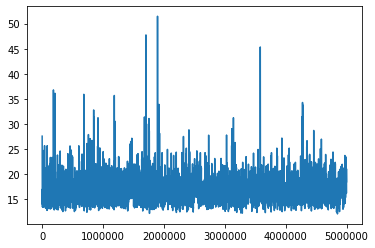}\includegraphics[width=0.45\linewidth]{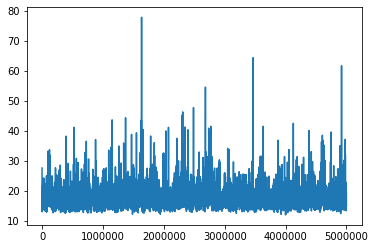}
\par\end{centering}
\begin{centering}
\includegraphics[width=0.45\linewidth]{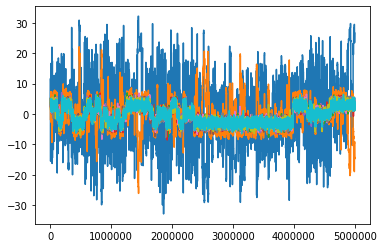}\includegraphics[width=0.45\linewidth]{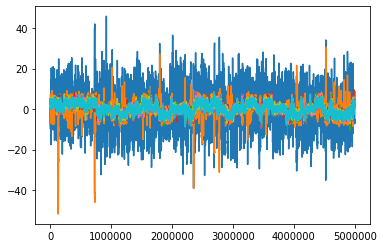}
\par\end{centering}
\begin{centering}
\includegraphics[width=0.45\linewidth]{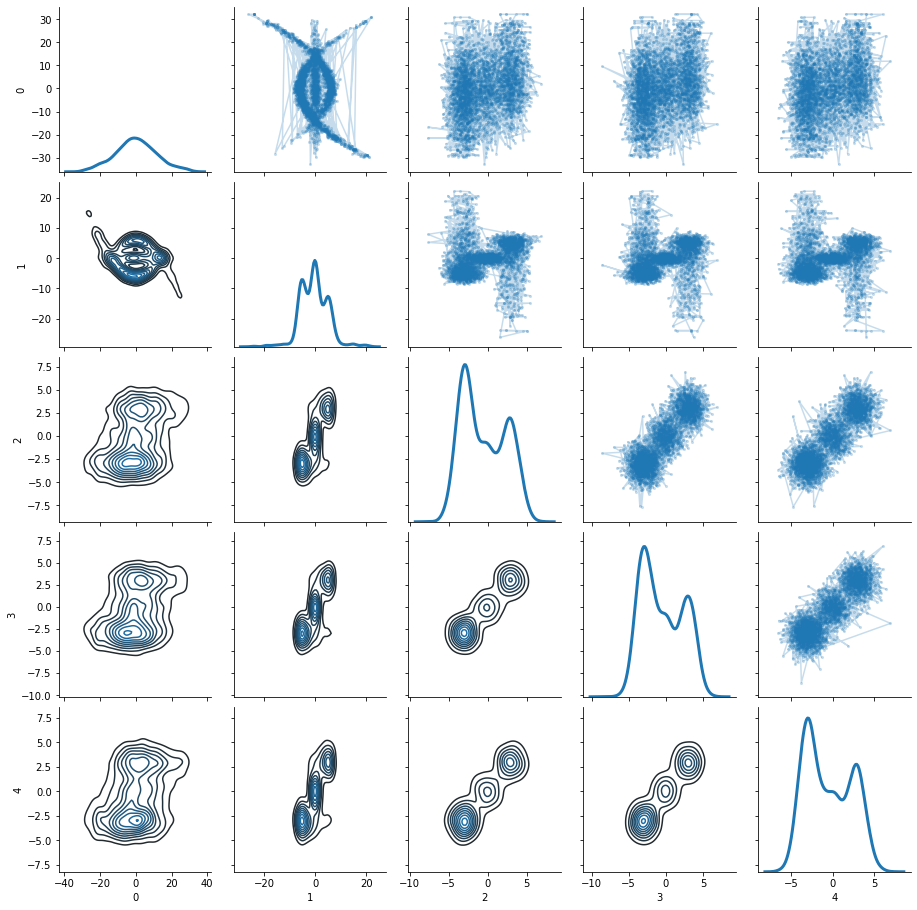}\includegraphics[width=0.45\linewidth]{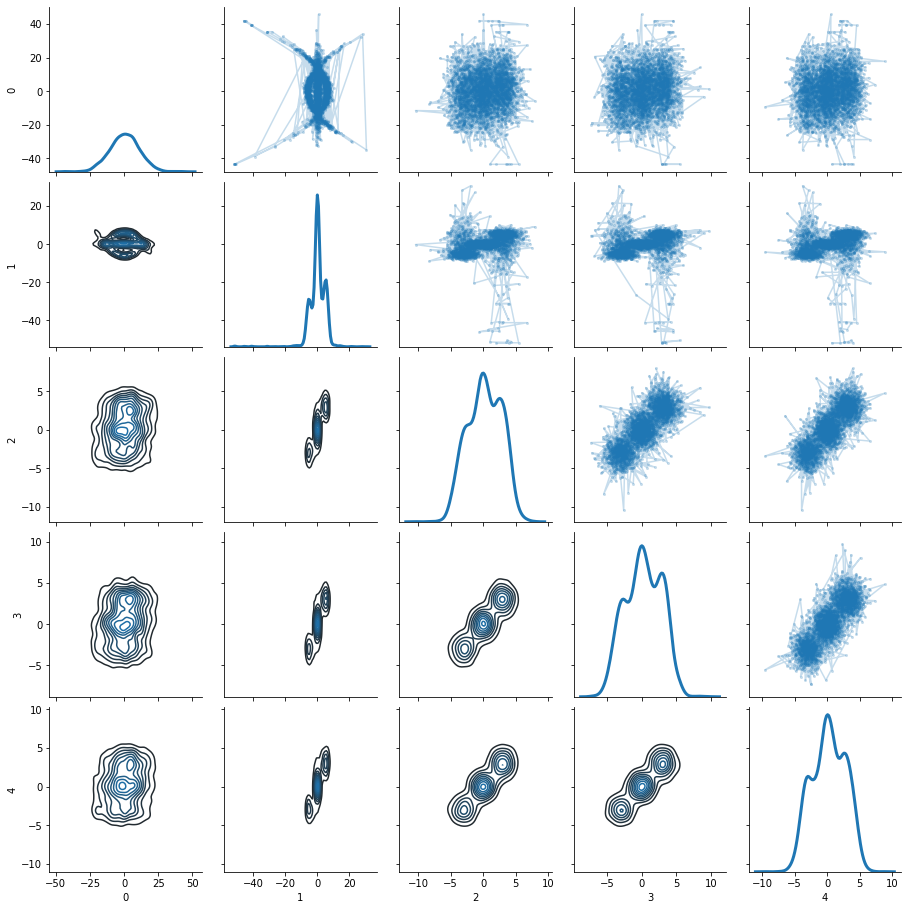}
\par\end{centering}
\caption{Plots from Example 4 for t-walk (left column plots)
and penalised t-walk (right column plots). Top-row figures correspond to trace plots for the value of $-\log(\pi)$. Middle-row plots are the traces for the first 5 components of the chain. Bottom-row figures correspond to the KDE plots (up to the first 5 components) estimated from the corresponding samples.}\label{fig:Fig4_10d3Mix}
\end{figure}

\section{Combining samples}\label{sec:combining}

One can always find a sufficiently pathological example in which modes are too isolated in order for our MCMC to visit them all. Even the penalised t-walk may not work correctly in such an example. This will be the case, for instance, if the two modes in Example 1 are placed extremely apart from each other.

As already mentioned, in the practise of MCMC one should always experiment with different starting values, to corroborate that our chain is not getting trapped in some regions of the sample space.  If it does get trapped, and within each separated region the chain appears to be mixing well, it becomes apparent that we have several modes and the chain cannot jump from one to the next.  How could one combine the two samples in order to obtain a new sample approximately distributed as $\pi$?
Certainly, one cannot simply mix the samples since this does not consider the probability accumulated in each region by the target $\pi$. Note that the discussion that follows considers any MCMC sampler, the t-walk, the penalised t-walk or otherwise.

Suppose that the target $\pi$ has ``significant'' mass on two disjoint
regions $\mathcal{X}_{1}$ and $\mathcal{X}_{2}$, and also assume
the MCMC sampler produces a $\varphi$-irreducible and aperiodic chain
but it has a very small probability of reaching $\mathcal{X}_{2}$
from $\mathcal{X}_{1}$ and viceversa. Sampling from $\pi$ relying
on a single chain $X_{0:T}:=\left(X_{t}\right)_{t=0}^{T}$ becomes
challenging since 
\begin{align*}
\mathbb{P}\left[X_{t}\in\mathcal{X}_{\text{2}}\text{ for some }t\in {1,\dots, T}\mid X_{0}\in\mathcal{X}_{1}\right] & \leq\varepsilon_{T},
\end{align*}
with $\varepsilon_{T}\ll1$. Despite the fact that $\varepsilon_{T}\rightarrow1$
as $T\rightarrow\infty$, in practice it is usually the case that
$X_{0:T}\approx\pi_{i}\left(\cdot\right)$ for some $i\in\left\{ 1,2\right\} $,
where
\begin{align*}
\pi_{i}\left(x\right) & \propto\pi\left(x\right)\mathbf{1}\left(x\in R_{i}\right).
\end{align*}

Additionally, notice that the target can be expressed as follows
\begin{align*}
\pi\left(x\right) & =\frac{\tilde{\pi}\left(x\right)}{Z}\\
 & =\frac{\tilde{\pi}\left(x\right)}{Z}\mathbf{1}\left(x\in\mathcal{X}_{1}\right)+\frac{\tilde{\pi}\left(x\right)}{Z}\mathbf{1}\left(x\in\mathcal{X}_{2}\right)\\
 & =\frac{\tilde{\pi}_{1}\left(x\right)}{Z_{1}}\left(\frac{Z_{1}}{Z}\right)+\frac{\tilde{\pi}_{2}\left(x\right)}{Z_{2}}\left(\frac{Z_{2}}{Z}\right)\\
 & =\pi_{1}\left(x\right)\frac{Z_{1}}{Z}+\pi_{2}\left(x\right)\frac{Z_{2}}{Z},
\end{align*}
where $Z_{i}=\int_{\mathcal{\mathcal{X}}_{i}}\tilde{\pi}\left(x\right)dx\leq Z$. 
Hence, suppose we have access to a sample from $\pi_{1}$ and one from
$\pi_{2}$, how could we obtain an approximate sample from $\pi$? We now discuss a solution that only requires unbiased estimates $Z_i^{-1}$.

\subsection{Asymmetric and pseudo-marginal MCMC}

Suppose we have two different proposals $q_{21}$ and $q_{12}$, acting
on the disjoint subsets $\mathcal{X}_{1}$ and $\mathcal{X}_{2}$,
respectively. A simple asymmetric MCMC kernel can be constructed as
follows
\begin{align*}
P\left(x,dy\right) & =P_{12}\left(x,dy\right)\mathbf{1}\left(x\in\mathcal{X}_{1}\right)+P_{21}\left(x,dy\right)\mathbf{1}\left(x\in\mathcal{X}_{2}\right),
\end{align*}
where the sub-kernel $P_{ij}\left(x,dy\right)=q_{ij}\left(x,dy\right)\alpha_{ij}\left(x,y\right)+\delta_{x}\left(dy\right)\rho_{ij}\left(x\right)$
considering
\begin{align*}
\alpha_{ij}\left(x,y\right) & =\min\left\{ 1,\frac{\pi\left(y\right)q_{ji}\left(y,x\right)}{\pi\left(x\right)q_{ij}\left(x,y\right)}\right\} ,
\end{align*}
and $\rho_{ij}\left(x\right)=1-\int_{\mathcal{X}}q_{ij}\left(x,dy\right)\alpha_{ij}\left(x,y\right)$.
It is not difficult to show that the kernel $P$ is invariant under
the target distribution $\pi$. More general implementations of this
idea are possible for which there are more than 2 proposals $\left\{ q_{ij}\right\} _{i,j}$
and also selecting sub-kernels at random \citep[see e.g][]{Tjelmeland2001,Andrieu2018}.

In our context, we consider two samples $X_{1:N}\sim\pi_{1}\left(\cdot\right)$
and $Y_{1:N}\sim\pi_{2}\left(\cdot\right)$,
and we would like to use the sub-targets $\left\{ \pi_{i}\right\} _{i}$
as proposals. Sampling is possible, however evaluating the corresponding densities is a challenge since the normalising constants
$Z_{i}=\int_{\mathcal{X}_{i}}\tilde{\pi}\left(dx\right)$ are not
known. The acceptance probability for this scenario, and when $x\in\mathcal{X}_{i}$,
becomes
\begin{align*}
\alpha_{ij}\left(x,y\right) & =\min\left\{ 1,\frac{\pi\left(y\right)\pi_{i}\left(x\right)}{\pi\left(x\right)\pi_{j}\left(y\right)}\right\} =\min\left\{ 1,\frac{Z_{i}^{-1}}{Z_{j}^{-1}}\mathbf{1}\left(x\in\mathcal{X}_{i}\right)\right\} .
\end{align*}
We require a procedure for estimating the ratio of normalising constants in
order to obtain an exact algorithm, if possible. One approach is to
estimate unbiasedly each term $Z_{i}^{-1}$ and resort to pseudo-marginal
(PM) theory \citep{Beaumont2003a,Andrieu2009}. Notice that for any density $h_{i}$ defined on $\mathcal{X}_{i}$
\begin{align*}
Z_{i}^{-1} & =\int_{\mathcal{X}_{i}}\frac{h_{i}\left(u\right)}{\tilde{\pi}\left(u\right)}\pi_{i}\left(du\right)\\
 & =\mathbb{E}_{U\sim\pi_{i}}\left[\frac{h_{i}\left(U\right)}{\tilde{\pi}\left(U\right)}\right].
\end{align*}
This type of estimators are not new in the literature \cite{Newton1994,Gelfand1994}, and have proved useful for estimating ratios of marginal likelihoods or Bayes factors, see e.g. \citep{Capistran2016b} and references therein. 

The resulting PM algorithm will target the extended distribution $\bar{\pi}$
defined on $\mathcal{X}^{2}$ with density
\begin{alignat*}{1}
\bar{\pi}\left(x,u\right) & =\pi\left(x\right)\left[h_{2}\left(u\right)\mathbf{1}\left(x\in\mathcal{X}_{1}\right)+h_{1}\left(u\right)\mathbf{1}\left(x\in\mathcal{X}_{2}\right)\right],
\end{alignat*}
proposing moves using $\bar{q}_{ij}\left(x,u;dy,dv\right)=\pi_{j}\left(dy\right)\pi_{i}\left(dv\right)$. Notice that the extended target $\bar{\pi}$ admits $\pi$ as one of its marginals since $\int \bar{\pi}(x,u)du=\pi(x)$, and the corresponding acceptance probability $\bar{\alpha}_{ij}$
(when $x\in\mathcal{X}_{i}$) simplifies to
\begin{align*}
\bar{\alpha}_{ij}\left(x,u;y,v\right) & =\min\left\{ 1,\frac{\bar{\pi}\left(y,v\right)\bar{q}_{ij}\left(y,v;x,u\right)}{\bar{\pi}\left(x,u\right)\bar{q}_{ij}\left(x,u;y,v\right)}\right\} \\
 & =\min\left\{ 1,\frac{h_{i}\left(v\right)/\tilde{\pi}\left(v\right)}{h_{j}\left(u\right)/\tilde{\pi}\left(u\right)}\right\} .
\end{align*}

Additionally, one could use multiple instances of the auxiliary variables
$U$ and $V$ discussed above, and modify the acceptance $\bar{\alpha}_{ij}$ accordingly for using ratios of arithmetic averages instead. It is well-known that doing this does not affect the exactness of the algorithm \citep{Beaumont2003a,Andrieu2009}, and in fact increasing the accuracy of the unbiased estimators within the acceptance ratio will typically reduce the asymptotic variance of the resulting chain \citep{Andrieu2015}.

Up to this point we have not discussed possible choices for $h_i$ and $h_j$. In principle they can be arbitrary, but ideally one would like $h_i$ to resemble $\pi$, restricted to $\mathcal{X}_i$. A feasible option, and as discussed in the following section, is to set $h_i$ equal to the resulting density from a KDE algorithm using a sample from $\pi_i$.

\subsubsection{Approximate sampling via KDE}

Using KDE one can construct an unbiased estimator for the reciprocal
of the normalising constants $Z_{1}$ and $Z_{2}$. Having
samples $X_{1:N}$ and $Y_{1:N}$ coming from $\pi_{1}$ and $\pi_{2}$,
respectively, if we construct an estimated kernel leaving out the $i$-th element
of the sample, say $X_{i}$, we obtain the approximate density $\widehat{\pi}_{1}^{(-i)}:\mathcal{X}_{1}\rightarrow\mathbb{R}_{0}^{+}$,
where the superscript $(-i)$ indicates we are not considering $X_{i}$.
Similarly using the sample $Y_{1:N}$ one can construct $\widehat{\pi}_{2}^{(-i)}$
for any $i\in\left\{ 1,\dots,N\right\} $. The reason for leaving
out a variable becomes clear when computing the following expectation using the Tower property:
\begin{align*}
\mathbb{E}\left[\frac{\widehat{\pi}_{1}^{(-i)}\left(X_{i}\right)}{\tilde{\pi}\left(X_{i}\right)}\right] & =\mathbb{E}\left[\mathbb{E}\left[\frac{\widehat{\pi}_{1}^{(-i)}\left(X_{i}\right)}{\tilde{\pi}\left(X_{i}\right)}\mid X_{(-i)}\right]\right]\\
 & =\mathbb{E}\left[\int_{\mathcal{X}_{1}}\frac{\widehat{\pi}_{1}^{(-i)}\left(x\right)}{\tilde{\pi}\left(x\right)}\pi_{1}\left(dx\right)\right]\\
 & =Z_{1}^{-1}\mathbb{E}\left[\int_{\mathcal{X}_{1}}\widehat{\pi}_{1}^{(-i)}\left(x\right)dx\right]\\
 & =Z_{1}^{-1}.
\end{align*}
Therefore, the ratio
\begin{align}
R_{j}^{(i)} & :=\frac{\widehat{\pi}_{j}^{(-i)}\left(X_{i}\right)}{\tilde{\pi}\left(X_{i}\right)}\label{eq:unbiasedRatios}
\end{align}
is an unbiased estimator of $Z_{j}^{-1}$ for any $i\in\{1,\dots,N\}$ and any $j\in\{1,2\}$.

Therefore, we will use ratios from \eqref{eq:unbiasedRatios} for approxi\-mating the pseudo-marginal process discussed in the previous section, meaning that $h_j=\widehat{\pi}_{j}^{(-i)}$ for some $i\in\left\{1,\dots,N\right\}$. Algorithm \ref{alg:JumpApprox} describes the process for obtaining an approximate draw from $\pi$, using the samples $X_{1:N}\sim \pi_1$ and $Y_{1:N}\sim \pi_2$. We must stress that this algorithm is not exact since the estimators $\left\{ R_{j}^{(i)} \right\}_{i,j}$ depend on the original samples $X_{1:N}$ and $Y_{1:N}$. 

Notice also that the algorithm produces a chain of indices $(M_t,I_t)_{t\geq 0}$, rather than a chain of values in $\mathcal{X}$. The first index ($M_t$) indicates which mode is being sampled, whereas the second one ($I_t$) corresponds to the variable sampled within $X_{1:N}$ (if $M_t=1$) or $Y_{1:N}$ (if $M_t=2$). With this in mind it is straightforward to obtain the actual chain of values, which is a sample approximately coming from $\pi$.

\begin{algorithm}
INPUT: Samples $x_{1:N}\leftarrow X_{1:N}\sim \pi_1$ and $y_{1:N}\leftarrow Y_{1:N}\sim \pi_2$. Ratios $r_{1}^{(1:N)}$ and $r_{2}^{(1:N)}$ computed using \eqref{eq:unbiasedRatios}. Current value for the chain $\left(m,i\right)\in \left\{ 1,2 \right\} \times \left\{ 1,\dots,N \right\}$.

OUTPUT: New value for the chain.

\begin{itemize}
\item If $m=1$ then:
\begin{enumerate}
    \item Set $n\leftarrow 2$;
    \item Sample and store $j\leftarrow J \sim Unif\left\{ 1,\dots,N\right\}$;
    \item Set $x\leftarrow X_i$ and  $y\leftarrow Y_j$;
    \item Return $\left(n,j\right)$ with probability $\bar{\alpha}_{1,2}^{N}\left(x,r_1^{(-i)};y,r_2^{(-j)}\right)$,
where 
\[
\bar{\alpha}_{1,2}^{N}\left(x,r_1^{(-i)};y,r_2^{(-j)}\right)=\min\left\{ 1, \frac{\sum_{k\neq i}^{N} r_1^{(k)} }{ \sum_{l\neq j}^{N} r_2^{(l)} }\right\} ;
\]

otherwise return $\left(m,i\right)$.
\end{enumerate}

\item Else:
\begin{enumerate}
    \item Set $n\leftarrow 1$;
    \item Sample and store $j\leftarrow J \sim Unif\left\{ 1,\dots,N\right\}$;
    \item Set $x\leftarrow Y_i$ and  $y\leftarrow X_j$;
    \item Return $\left(n,j\right)$ with probability $\bar{\alpha}_{2,1}^{N}\left(x,r_2^{(-i)};y,r_1^{(-j)}\right)$,
where 
\[
\bar{\alpha}_{2,1}^{N}\left(x,r_2^{(-i)};y,r_1^{(-j)}\right)=\min\left\{ 1, \frac{\sum_{k\neq i}^{N} r_2^{(k)} }{ \sum_{l\neq j}^{N} r_1^{(l)} }\right\} ;
\]

otherwise return $\left(m,i\right)$.
\end{enumerate}
\end{itemize}
\caption{Jumping modes using KDE}\label{alg:JumpApprox}
\end{algorithm}

\begin{example}
We consider a modified version of Example 1, where the two components are still the same but now the weights are different, i.e.
\begin{align*}
\pi\left(x_{1},x_{2}\right) & =0.1\mathcal{N}\left(\left(x_{1},x_{2}\right)\mid\mu_{\text{1}},\Sigma_{1}\right)+0.9\mathcal{N}\left(\left(x_{1},x_{2}\right)\mid\mu_{\text{2}},\Sigma_{2}\right).
\end{align*}

In this case we set $\pi_i=\mathcal{N}\left(\mu_{\text{i}},\Sigma_{i}\right)$ and draw 10 thousand samples from each distribution. Using Algorithm \ref{alg:JumpApprox} we aim at combining these samples that lie on well-separated regions of the state space. Figure \ref{fig:Fig5_combine} compares the KDE plots resulting from the proposed method (left) and a simple resampling process with the correct weights (right). The jumping modes algorithm was run for 100 thousand iterations showing a global acceptance of 20 percent. Notice that the resulting plots are almost identical; hence, even though the proposed algorithm is not exact it may prove useful for post-processing samples.
\end{example}

\begin{figure}[!ht]
\begin{centering}
\includegraphics[width=0.45\linewidth]{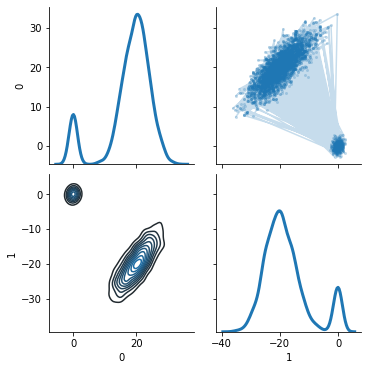}\includegraphics[width=0.45\linewidth]{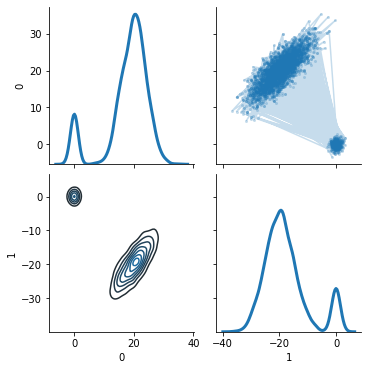}
\par\end{centering}
\caption{KDE plots from Example 5 for the jumping modes algorithm  (left) and a resampling process using the correct weights (right).}\label{fig:Fig5_combine}
\end{figure}

\section{Final discussion}\label{sec:discussion}

The introduction of penalised moves into the t-walk algorithm can be useful for tackling  multimodality in complex inference scenarios. In principle one could implement the aforementioned moves, either gradient-based or non-gradient based, into other MCMC samplers. In this respect, modifying the penalty in adaptive manner may sound appealing, although one must be careful of not breaking the ergodicity of the chain.

The penalised t-walk works well in various examples, but as observed, the algorithm may still struggle to jump between modes in complex or high-dimensional settings. This is due to chances of finding a new mode, hence accepting a penalised move, may become quite small. For such cases, we presented an approach for combining samples from two different samplers, that in theory should be able to visit all modes, but realistically they will get stuck in two separate regions of the state space.

We want to emphasise that combining samples, as presented here, should be regarded as a last resort since the validity relies on strong assumptions. For instance, the two samples should be of decent quality, meaning that the corresponding MCMC samplers may need to run for several iterations. Additionally, the samples need to lie on disjoint, (or almost disjoint) regions of the state space, which may be difficult to asses in higher dimensions. Nevertheless, extending the method for non-disjoint sets and for more than two samples should be possible, but has been postponed for future research.

\section{Acknowledgments}

The authors were partially founded by CONACYT CB-2016-01-284451 and COVID19 312772 grants, and a RDCOMM grant.


\end{document}